\documentclass[letterpaper,onecolumn,accepted=2020-02-24]{quantumarticle}

\pdfoutput=1

\usepackage{eucal}
\usepackage{dsfont}
\usepackage{amsmath}
\usepackage{amssymb}
\usepackage{amsthm}
\usepackage{mathtools}
\usepackage{tikz} \usetikzlibrary{calc,shapes}
\usepackage[numbers]{natbib}
\usepackage{hyperref}

\usepackage{longtable}

\usepackage{silence}
\newtheorem{theorem}{Theorem}
\newtheorem{lemma}[theorem]{Lemma}

\theoremstyle{definition}
\newtheorem{definition}[theorem]{Definition}
\newtheorem{remark}[theorem]{Remark}

\newcommand{\abs}[1]{\lvert #1 \rvert}
\newcommand{\bigabs}[1]{\bigl\lvert #1 \bigr\rvert}
\newcommand{\Bigabs}[1]{\Bigl\lvert #1 \Bigr\rvert}

\newcommand{\ip}[2]{\langle #1 , #2\rangle}
\newcommand{\bigip}[2]{\bigl\langle #1, #2 \bigr\rangle}

\newcommand{\biggip}[2]{\biggl\langle #1, #2 \biggr\rangle}

\newcommand\complex{\mathbb{C}}
\newcommand\real{\mathbb{R}}

\newcommand{\tinyspace}{\mspace{1mu}}

\newcommand{\op}[1]{\operatorname{#1}}
\newcommand{\tr}{\operatorname{Tr}}

\renewcommand{\vec}{\operatorname{vec}}

\newcommand{\norm}[1]{\lVert\tinyspace #1 \tinyspace\rVert}
\newcommand{\bignorm}[1]{\bigl\lVert\tinyspace #1 \tinyspace\bigr\rVert}

\newcommand{\biggnorm}[1]{\biggl\lVert\tinyspace #1 \tinyspace\biggr\rVert}

\newcommand\I{\mathds{1}}

\newcommand{\setft}[1]{\mathrm{#1}}

\newcommand{\Unitary}{\setft{U}}
\newcommand{\MixedUnitary}{\setft{MU}}
\newcommand{\Herm}{\setft{Herm}}
\newcommand{\Lin}{\setft{L}}
\newcommand{\Channel}{\setft{C}}

\newcommand\A{\mathcal{A}}
\newcommand\B{\mathcal{B}}
\newcommand\V{\mathcal{V}}

\newcommand\K{\mathcal{K}}
                 
\definecolor{White}{rgb}{1,1,1}
\definecolor{Black}{rgb}{0,0,0}
\definecolor{LightGray}{rgb}{.81,.81,.81}
\colorlet{ChannelColor}{LightGray}
\colorlet{ChannelTextColor}{Black}
\colorlet{ReadoutColor}{White}
 
\begin{document}

\title{Detecting mixed-unitary quantum channels is NP-hard}
\author{Colin Do-Yan Lee}
\author{John Watrous}

\affil{%
  Institute for Quantum Computing and School of Computer Science\protect\\
  University of Waterloo, Canada\vspace{2mm}}

\date{February 8, 2019}

\maketitle

\begin{abstract}
  A quantum channel is said to be a \emph{mixed-unitary} channel if it can be
  expressed as a convex combination of unitary channels.
  We prove that, given the Choi representation of a quantum channel $\Phi$,
  it is NP-hard with respect to polynomial-time Turing reductions to determine
  whether or not $\Phi$ is a mixed-unitary channel.
  This hardness result holds even under the assumption that $\Phi$ is not
  within an inverse-polynomial distance (in the dimension of the space upon
  which $\Phi$ acts) of the boundary of the mixed-unitary channels.
\end{abstract}

\section{Introduction}

In the theory of quantum information, \emph{quantum channels} represent
discrete-time changes in systems that can, in an idealized sense, be realized
by physical processes.
Mathematically speaking, quantum channels are represented by completely
positive and trace-preserving linear maps of the form
\mbox{$\Phi:\Lin(\complex^n) \rightarrow \Lin(\complex^m)$},
where $\Lin(\complex^n)$ is the set of linear maps, or \emph{operators}, from
$\complex^n$ to itself, and likewise for $\Lin(\complex^m)$.
If the state of a system is represented by a density operator
$\rho\in\Lin(\complex^n)$ prior to the action represented by the channel
$\Phi$, then its state after the channel acts is given by the density operator
$\Phi(\rho) \in \Lin(\complex^m)$.
This paper focuses on channels for which $n=m$, which represent the common
situation in which a discrete-time change preserves the size of a physical
system.
(The sizes of the input and output systems of a quantum channel are reflected
by the dimensions of the underlying spaces $\complex^n$ and $\complex^m$.)

\emph{Unitary channels} form one of the simplest categories of quantum
channels.
A unitary channel is a channel of the form
$\Phi:\Lin(\complex^n) \rightarrow \Lin(\complex^n)$
that is given by $\Phi(X) = U X U^{\ast}$ for every $X\in\Lin(\complex^n)$, for
some fixed choice of a unitary operator $U\in\Lin(\complex^n)$.
A \emph{mixed-unitary channel} is one that can be expressed as a convex
combination of unitary channels.
Equivalently, a channel $\Phi:\Lin(\complex^n) \rightarrow \Lin(\complex^n)$
is mixed-unitary if there exists a positive integer $N$, a probability vector
$(p_1,\ldots,p_N)$, and unitary operators $U_1,\ldots,U_N\in\Lin(\complex^n)$
such that
\begin{equation}
  \Phi(X) = \sum_{k=1}^N p_k U_k X U_k^{\ast}
\end{equation}
for every operator $X\in\Lin(\complex^n)$.
We let $\text{MU}(\complex^n)$ denote the set of all such channels.

Mixed-unitary channels are important in quantum information theory for a number
of reasons.
They provide a rich set of examples of channels, and exhibit many fundamental
attributes and properties of general quantum channels \cite{Rosgen08}.
At the same time, their relatively simple form can be beneficial in analyses,
as compared with general quantum channels.
Mixed-unitary channels arise naturally in both Hermitian operator formulations
of majorization \cite{AlbertiU82} and in a variety of cryptography situations
that concern the encryption of quantum states
\cite{AmbainisMTdW00,HaydenLSW04,AmbainisS04}.

Quantum channels can be represented in different ways, but one common
representation is the \emph{Choi representation} \cite{Choi75}.
The Choi representation of an arbitrary linear map
$\Phi:\Lin(\complex^n)\rightarrow\Lin(\complex^n)$ is defined as
\begin{equation}
  J(\Phi) = \sum_{1\leq i,j\leq n} \Phi(E_{i,j}) \otimes E_{i,j},
\end{equation}
where $E_{i,j}$ is the operator mapping the elementary unit vector $e_j$ to
$e_i$ and all vectors orthogonal to $e_i$ to $0$.
(Equivalently, with respect to the standard basis $\{e_1,\ldots,e_n\}$, the
operator $E_{i,j}$ is represented by the matrix having a 1 in entry $(i,j)$ and
all other entries 0.)

We prove that it is NP-hard, with respect to polynomial-time Turing reductions,
to determine whether or not a given quantum channel is mixed-unitary.
Specifically, we consider the problem in which the input is
the Choi representation $J(\Phi) \in \Lin(\complex^n \otimes \complex^n)$ of a
quantum channel $\Phi:\Lin(\complex^n)\rightarrow\Lin(\complex^n)$, along with
the unary representation $0^m$ of a positive integer $m$, and the task is to
determine whether or not $\Phi$ is a mixed-unitary channel under the promise
that $J(\Phi)$ is not within distance $1/m$ of the boundary of the
set of all Choi representations of mixed-unitary channels.
That is, the promise guarantees that the decision of whether or not $\Phi$ is
mixed-unitary is not ``artificially hard'' due to issues relating to
numerical precision.
Our proof establishes that this problem is, in fact, \emph{strongly NP-hard},
meaning that it remains NP-hard even when the real and imaginary parts of all
of the numbers appearing in the Choi representation of the input channel are
expressible as ratios of integers that are bounded in absolute value by a
polynomial in the length of the entire input.

The methodology behind our proof is reminiscent of known proofs of the
NP-hardness of testing if a given bipartite density operator
$\rho\in\Lin(\complex^n\otimes\complex^m)$ is \emph{separable}
\cite{Gurvits2003,Ioannou2007,Gharibian2010,Yu2016,Tura2018}, meaning that it can be
represented as a convex combination of \emph{product states}, which represent
independence between the two individual systems that define the bipartition in
question.
In particular, following the strong NP-hardness proof of separability testing
due to Gharibian \cite{Gharibian2010}, we make use of a theorem due to Liu
\cite{Liu2007} that establishes the existence of a polynomial-time Turing
reduction from the \emph{weak optimization} problem to the
\emph{weak membership} problem in certain families of convex sets.
We note that our main result can, in fact, be closely linked with the problem
of separability testing, in the sense that testing if a channel is
mixed-unitary may alternatively be formulated as a problem concerning the
expression of a bipartite density operator in a certain way.
More specifically, the set of all Choi representations of mixed-unitary
channels, when normalized, is equivalent to the set of bipartite quantum states
that can be written as convex combinations of \emph{maximally entangled}
states.

The remainder of this paper is organized as follows.
Section~\ref{sec:preliminaries} summarized preliminary material on
computational complexity and quantum information theory that is required to
understand this paper, and formally defines the \emph{mixed-unitary detection}
problem described above along with a different problem, called
\emph{unitary quadratic minimization}, that plays an important role in the
proof of our main result.
In Section~\ref{sec:reduction-from-3-coloring} we prove that the NP-complete
\emph{graph 3-coloring} problem (3COL) reduces to
\emph{unitary quadratic minimization} (UQM) through a polynomial-time mapping
reduction, and in Section~\ref{sec:reduction-to-MUD} we prove that
\emph{unitary quadratic minimization} reduces to \emph{mixed-unitary detection}
(MUD) through a polynomial-time Turing reduction.
In symbols, these two sections establish the relations
\begin{equation}
  \mathrm{3COL} \leq_m^p \mathrm{UQM} \leq_T^p \mathrm{MUD},
\end{equation}
which implies the NP-hardness of testing if a channel is
mixed-unitary.
The paper concludes with Section~\ref{sec:conclusion}, which mentions a few
open problems that relate to the main results of the paper.

\section{Preliminaries}
\label{sec:preliminaries}

The main purpose of this section is to clarify some of the notation and
conventions we use throughout the paper, and to define two decision problems:
one is the \emph{mixed-unitary detection} problem, whose hardness is the
primary focus of this paper, and the second is the
\emph{unitary quadratic minimization} problem, which serves as an intermediate
problem through which an NP-complete problem (the \emph{graph 3-coloring}
problem) is reduced to the \emph{mixed-unitary detection} problem.

\subsection{Computational complexity}

We assume the reader is familiar with basic notions of computational
complexity, such as polynomial-time mapping reductions, polynomial-time
Turing reductions, the concept of NP-completeness, and the fact that the
\emph{graph 3-coloring} problem is NP-complete.
This material is covered in several textbooks on computational complexity,
such as the book of Arora and Barak \cite{AroraB09}.
When we speak of \emph{polynomials} in this paper, we are referring only
to resource bounds---so it should be understood that we are referring more
precisely to nonzero univariate polynomials having non-negative integer
coefficients.

The decision problems we consider involve approximations of real number values
and/or guarantees on distances between real or complex vectors, and for this
reason they are naturally stated as \emph{promise problems}
\cite{EvenSY84}.
Formally speaking, a promise problem is a pair
$A = (A_{\text{yes}},A_{\text{no}})$ of disjoint sets 
$A_{\text{yes}},A_{\text{no}}\subseteq\Sigma^{\ast}$ of strings over an
alphabet $\Sigma$.
A hypothetical algorithm or protocol for $A$ is required to output ``yes''
(or 1) on input strings in $A_{\text{yes}}$ (which are called \emph{yes-inputs}
or \emph{yes-instances}) and output ``no'' (or 0) on input strings in
$A_{\text{no}}$ (which are called \emph{no-inputs} or \emph{no-instances}).
No constraints are placed on an algorithm or protocol for $A$ on strings
outside of the set $A_{\text{yes}}\cup A_{\text{no}}$.
In the promise problem statements found below and later in the paper, we
first list general assumptions on the form of the input, which is understood to
be a string encoding of one or more mathematical objects, followed by a
specification of which of these inputs are to be considered yes-instances and
which are to be considered no-instances.

In the problems considered in this paper, every complex number is assumed to
be encoded as a triple $(x,y,z)$ that represents the number $(x + i y)/z$,
where $x$ and $y$ are integers represented in signed binary notation and $z$ is
a positive integer represented in binary notation.
Real numbers are encoded similarly, but where the imaginary part represented by
$y$ is omitted.
One exception is when positive integers are explicitly stated to be represented
in unary notation, which means that each positive integer $m$ is encoded as the
string $0^m$.
Real or complex vectors and matrices are encoded as complete lists of their
real or complex number entries (as opposed to compact representations of sparse
matrices, for instance).

For a given polynomial $p$, we may say that an instance of any of the problems
discussed in this paper is \emph{$p$-bounded} if, for every real or complex
number appearing in that problem instance (and encoded as described above), the
values $x$, $y$, and $z$ are bounded in absolute value by $p(n)$, for $n$ being
the length of the entire instance being considered.
A polynomial-time mapping reduction $A\leq_m^p B$ between promise problems $A$
and $B$ will be called \emph{strong} if, for every polynomial $p$ there exists
a polynomial $q$ such that this property holds:
for every $p$-bounded instance of $A_{\text{yes}}$ or $A_{\text{no}}$,
the reduction produces a $q$-bounded instance of $B_{\text{yes}}$ or
$B_{\text{no}}$, respectively.
Along similar lines, a polynomial-time Turing reduction $A\leq_T^p B$ is
\emph{strong} if, for every polynomial $p$ there exists a polynomial $q$ such
that, on every $p$-bounded instance of $A_{\text{yes}}$ or $A_{\text{no}}$, the
reduction only queries $q$-bounded instances of $B$, and accepts or rejects
accordingly.
Finally, a problem is \emph{strongly NP-hard} (with respect to either
polynomial-time mapping or Turing reductions) if it remains NP-hard even under
the additional promise that every yes- or no-instance is $p$-bounded, for some
choice of a polynomial $p$.

\subsection{Linear algebra and quantum information}

Similar to computational complexity, we assume that the reader is familiar with
basic notions of linear algebra and quantum information.
There is, in fact, little in the way of quantum information theory that is
required for an understanding of this paper, aside from the definition of
quantum channels and their Choi representations (which are described in Chapter
2 of \cite{Watrous2018}, for instance).

For $n$ a positive integer, the vector space $\complex^n$ is defined in the
usual way, an inner product on this space is defined as
\begin{equation}
  \ip{u}{v} = \sum_{k=1}^n \overline{u(k)} v(k)
\end{equation}
(conjugate linear in the first argument), and the \emph{Euclidean norm} is
given by
\begin{equation}
  \norm{u} = \sqrt{\ip{u}{u}}.
\end{equation}
The \emph{standard basis} of $\complex^n$ is the basis
$\{e_1,\ldots,e_n\}$ of elementary unit vectors.

We write $\Lin(\complex^n)$ to denote the set of linear operators (or mappings)
from $\complex^n$ to itself, and associate this set with the set of all
$n\times n$ complex matrices, where the understanding is that the matrix
is a representation of the operator with respect to the standard basis.
We already introduced the notation $E_{i,j}$ in the previous section; the
operator $E_{i,j}$ is the operator whose matrix representation has a 1 in entry
$(i,j)$ and 0 in all other entries.
The inner product of two operators $A,B\in\Lin(\complex^n)$ is defined as
$\ip{A}{B} = \tr(A^{\ast} B)$,
where $A^{\ast}$ is the \emph{adjoint} of $A$ (which, in terms of matrix
representations, is equivalent to the \emph{conjugate transpose} of $A$).
An operator $A\in\Lin(\complex^n)$ is \emph{Hermitian} if $A = A^{\ast}$, and
is \emph{unitary} if $A^{\ast} A = \I_n$, where $\I_n \in \Lin(\complex^n)$ is
the identity operator acting on $\complex^n$.
The notations $\Herm(\complex^n)$ and $\Unitary(\complex^n)$ refer to the sets
of all Hermitian and unitary operators in $\Lin(\complex^n)$, respectively.

We refer to three different norms of operators.
The \emph{spectral norm} of $A$ is defined as
\begin{equation}
  \norm{A} = \max\bigl\{\norm{Au}\,:\,u\in\complex^n,\;\norm{u} \leq 1\bigr\},
\end{equation}
the \emph{2-norm} (or \emph{Frobenius norm}) of $A$ is defined as
\begin{equation}
  \norm{A}_2 = \sqrt{\ip{A}{A}},
\end{equation}
and the \emph{trace norm} is defined as
\begin{equation}
  \norm{A}_1 = \tr \Bigl( \sqrt{A^{\ast} A} \Bigr),
\end{equation}
where $\sqrt{A^{\ast} A}$ is the unique positive semidefinite operator whose
square is $A^{\ast} A$.
These norms satisfy $\norm{A}\leq\norm{A}_2\leq \norm{A}_1$ for every
$A\in\Lin(\complex^n)$.

The notion of quantum channels was also already introduced in the previous
section.
For the purposes of this paper, it suffices to note that the property of
\emph{complete positivity} of a linear map
$\Phi:\Lin(\complex^n)\rightarrow\Lin(\complex^n)$ is equivalent to its Choi
representation
\begin{equation}
  J(\Phi) = \sum_{1\leq i,j\leq n} \Phi(E_{i,j}) \otimes E_{i,j}
\end{equation}
being positive semidefinite.
Similarly, the property that $\Phi$ is \emph{Hermitian-preserving}
(which means that $\Phi(X) \in \Herm(\complex^n)$ for every
$X\in\Herm(\complex^n)$)
is equivalent to $J(\Phi)$ being Hermitian.
The property that $\Phi$ \emph{preserves trace} is equivalent to
\begin{equation}
  \bigl(\textup{Tr}\otimes\I_{\Lin(\complex^n)}\bigr)(J(\Phi)) = \I_n,
\end{equation}
and the property that $\Phi$ is \emph{unital} (which means that
$\Phi(\I_n) = \I_n$) is equivalent to 
\begin{equation}
  \bigl(\I_{\Lin(\complex^n)}\otimes\textup{Tr}\bigr)(J(\Phi)) = \I_n,
\end{equation}
where $\I_{\Lin(\complex^n)}$ refers to the identity mapping from
$\Lin(\complex^n)$ to itself.

\subsection{Problem statements}

Finally, we formally define the decision problems (stated as promise problems)
that were referred to in the introduction.
\pagebreak

\begin{definition}
  The \emph{unitary quadratic minimization} (UQM) promise problem is as
  follows.\vspace{2mm}

  \noindent
  \begin{longtable}{@{}lp{5.6in}@{}}
    Input: & Operators $A_1,\ldots,A_k \in \Lin(\complex^n)$ with
    $\norm{A_j}_2 \leq 1$ for each $j\in\{1,\ldots,k\}$, a real number
    $\alpha$, and the unary representation $0^m$ of a positive
    integer~$m$.\\[2mm]
    Yes: & There exists a unitary operator $U \in \Unitary(\complex^n)$ such
    that
    \begin{equation}
      \sum_{j = 1}^k \abs{\ip{A_j}{U}}^2 \leq \alpha.
    \end{equation}\\
    No: & For every unitary operator $U \in \Unitary(\complex^n)$ it is the case
    that
    \begin{equation}
      \sum_{j = 1}^k \abs{\ip{A_j}{U}}^2 \geq \alpha + \frac{1}{m}.
    \end{equation}
  \end{longtable}
\end{definition}

\begin{definition}
  \label{definition:mixed-unitary detection}
  The \emph{mixed-unitary detection} (MUD) promise problem is as
  follows.\vspace{2mm}

  \noindent
  \begin{longtable}{@{}lp{5.6in}@{}}
    Input: & The Choi representation
    $J(\Phi)\in\Herm(\complex^n\otimes\complex^n)$ of a trace-preserving,
    unital, and Hermitian-preserving map
    $\Phi:\Lin(\complex^n)\rightarrow\Lin(\complex^n)$ along with the unary
    representation $0^m$ of a positive integer~$m$.\\[2mm]
    Yes: & Every trace-preserving, unital, Hermitian-preserving
    map $\Psi:\Lin(\complex^n)\rightarrow \Lin(\complex^n)$ that satisfies
    \begin{equation}
      \bignorm{J(\Psi) - J(\Phi)}_2 \leq \frac{1}{m}
    \end{equation}
    is a mixed-unitary channel.\\[2mm]
    No: & Every trace-preserving, unital, Hermitian-preserving
    map $\Psi:\Lin(\complex^n)\rightarrow \Lin(\complex^n)$ that satisfies
    \begin{equation}
      \bignorm{J(\Psi) - J(\Phi)}_2 \leq \frac{1}{m}
    \end{equation}
    it not a mixed-unitary channel.
  \end{longtable}
\end{definition}

Two brief remarks about the previous two definitions are in order.
First, the fact that both problems expect a \emph{unary} representation of a
positive integer $m$ (as opposed to a binary representation, say) is a
standard mechanism in computational complexity that forces the value of $m$ to
be polynomially related to its input size.
If the integer $m$ were instead to be input in binary notation, the value
of $m$ would be exponential in its input length, and as a result the problems
themselves would become harder in a computational sense.
As our main result is a hardness result, it is therefore a stronger result
given the assumption that $m$ is input in unary notation.

Second, we note that in Definition~\ref{definition:mixed-unitary detection}
specifically, the maps $\Phi$ and $\Psi$ are assumed to range over all
trace-preserving, unital, and Hermitian-preserving maps rather than over all
unital channels.
We have defined the mixed-unitary detection problem in this way so that
it has a form that is standard within the study of computational geometric
problems \cite{GroetschelLS1988}:
the input is a vector in a real vector space, and the task is to determine
if a ball of a certain radius either lies entirely within or is disjoint from a
set of interest.
The computational difficulty of the problem would not change in a fundamental
way if $\Phi$ were assumed to be a channel, given that this property can be
tested efficiently.
We do, however, require that $\Psi$ ranges over all trace-preserving, unital,
and Hermitian-preserving maps in the yes case, so as to properly reflect the
property that every map in a ball of radius $1/m$ around $\Phi$ is a
mixed-unitary channel.

\section{Reduction from graph 3-coloring to unitary quadratic optimization}
\label{sec:reduction-from-3-coloring}
 
In this section we prove that the \emph{unitary quadratic minimization}
problem is NP-hard, via a polynomial-time mapping reduction from the
\emph{graph 3-coloring} problem.
Our reduction establishes that this problem is, in fact, strongly NP-hard, as
the operators $A_1,\ldots,A_k$ that are produced by our reduction from any
instance of \emph{graph 3-coloring} have entries restricted to the set
$\{0,1/2,1\}$.

\subsection{The reduction}

Let $G = (V,E)$ be a graph with $n$ vertices $V = \{1,\ldots,n\}$ and
$m$ edges
\begin{equation}
  E = \bigl\{\{a_1,b_1\},\ldots,\{a_m,b_m\} \bigr\},
\end{equation}
where $1 \leq a_j < b_j \leq n$ for every $j \in \{1,\ldots,m\}$.
Let $N = n + m$ and consider the following two collections of operators
drawn from $\Lin(\complex^N)$:
\begin{enumerate}
\item[1.]
  $E_{i,j}$ for every choice of $i,j\in\{1,\ldots,N\}$ with $i\not=j$.
\item[2.]
  $\bigl(E_{a_j,a_j} + E_{b_j,b_j} + E_{n+j,n+j}\bigr)/2$ for every
  $j\in\{1,\ldots,m\}$.
\end{enumerate}
(The factor of $1/2$ in the second type of operator guarantees that
each of the operators produced by the reduction has 2-norm at most 1.)
The total number of operators in these two collections is
\begin{equation}
  k = (N^2 - N) + m = (n+m)^2 - n.
\end{equation}
Let $A_1,\ldots,A_k$ denote these operators taken in any reasonable ordering
that allows for the computation of these operators in polynomial time given the
graph $G$.
The instance of $\text{UQM}$ produced by the reduction is
\begin{equation}
  \Bigl( A_1,\ldots,A_k,0,0^{526n^2} \Bigr)
\end{equation}

\subsection{Analysis: yes-instances map to yes-instances}

Assume first that $G$ is 3-colorable, so that there exists a function
$\varphi:\{1,\ldots,n\}\rightarrow\{0,1,2\}$ with the property that
$\varphi(a)\not=\varphi(b)$ whenever $\{a,b\}\in E$.
One may obtain a unitary operator $U\in\Unitary(\complex^N)$ such that
\begin{equation}
  \label{eq:objective-value-zero}
  \sum_{j = 1}^k \abs{\ip{A_j}{U}}^2 = 0
\end{equation}
by taking $U$ to be the diagonal operator whose diagonal entries are
third-roots of unity as follows (assuming $\omega = \exp(2\pi i/3)$):
\begin{enumerate}
\item[1.]
  For each $a\in\{1,\ldots,n\}$, let $U(a,a) = \omega^{\varphi(a)}$.
\item[2.]
  For each $j \in \{1,\ldots,m\}$, let $U(n+j,n+j) = \omega^{c_j}$
  for $c_j \in \{0,1,2\}$ being the unique color such that
  $c_j\not\in\{\varphi(a_j),\varphi(b_j)\}$.
\end{enumerate}
As $U$ is diagonal, it is the case that $\bigip{E_{i,j}}{U} = 0$ whenever
$i\not=j$.
For the $j$-th edge $\{a_j,b_j\}$, it is the case that
\begin{equation}
  \bigip{\bigl(E_{a_j,a_j} + E_{b_j,b_j} + E_{n+j,n+j}\bigr)/2}{U}
  = \frac{1}{2}\bigl(U(a_j,a_j) + U(b_j,b_j) + U(n+j,n+j)\bigr),
\end{equation}
which is zero because it is proportional to the sum of the three roots of unity
$1 = \omega^0$, $\omega^1$, and $\omega^2$.

\subsection{Analysis: no-instances map to no-instances}

It remains to prove that if $G$ is not 3-colorable, then for every unitary
operator $U$ one has
\begin{equation}
  \label{eq:objective-value-big}
  \sum_{j = 1}^k \abs{\ip{A_j}{U}}^2 \geq \frac{1}{526 n^2}.
\end{equation}
This statement will be proved in the contrapositive form.
To this end, assume hereafter that $U$ is a unitary operator, and for the
operators $A_1,\ldots,A_k$ produced from a given graph $G$ by the reduction
described previously it is the case that
\begin{equation}
  \label{eq:objective-value-not-big}
  \sum_{j = 1}^k \abs{\ip{A_j}{U}}^2 < \eta = \frac{1}{526 n^2}.
\end{equation}
From this assumption we will recover a 3-coloring of the graph $G$.
We will make use of the following lemma, which is proved at the end of the
present subsection, to do this.

\begin{lemma}
  \label{lemma:angle-bound}
  Suppose $\varepsilon \in [0,1/6]$ and $\alpha,\beta,\gamma\in\complex$
  satisfy the following conditions:
  \begin{enumerate}
  \item[1.]
    $\abs{\alpha}, \abs{\beta}, \abs{\gamma}\in[1 - \varepsilon, 1]$.
  \item[2.]
    $\abs{\alpha + \beta + \gamma} \leq \varepsilon$.
  \end{enumerate}
  For each angle
  $\theta \in \{\arg(\alpha) - \arg(\beta),\,
  \arg(\beta) - \arg(\gamma),\,\arg(\gamma) - \arg(\alpha)\}$,
  interpreted as an element of the set $[0,2\pi)$, it is the case that
  \begin{equation}
    \Bigabs{\theta-\frac{2\pi}{3}} \leq 6\varepsilon
    \quad\text{or}\quad
    \Bigabs{\theta-\frac{4\pi}{3}} \leq 6\varepsilon.
  \end{equation}
\end{lemma}

We begin by observing that the diagonal entries of $U$ must be close to 1 in
absolute value.
Specifically, for every $j\in\{1,\ldots,N\}$ it is the case that
\begin{equation}
  \sum_{i\not=j} \abs{U(i,j)}^2
  = \sum_{i\not=j} \abs{\ip{E_{i,j}}{U}}^2
  \leq \sum_{i = 1}^k \abs{\ip{A_i}{U}}^2 < \eta,
\end{equation}
and therefore
\begin{equation}
  \abs{U(j,j)} > \sqrt{1 - \eta} \geq 1 - \sqrt{\eta},
\end{equation}
as every column of $U$ has unit norm.
Next, observe that
\begin{equation}
  \begin{multlined}
    \frac{1}{2}\abs{U(a_j,a_j) + U(b_j,b_j) + U(n+j,n+j)}\\
    = \frac{1}{2}\bigabs{\bigip{E_{a_j,a_j} + E_{b_j,b_j} + E_{n+j,n+j}}{U}}
    \leq
    \Biggl(\sum_{i = 1}^k \abs{\ip{A_i}{U}}^2\Biggr)^{\frac{1}{2}}
    \leq \sqrt{\eta}
  \end{multlined}
\end{equation}
for every $j \in \{1,\ldots,m\}$.
By Lemma~\ref{lemma:angle-bound}, we conclude that for any two adjacent
vertices $a,b\in\{1,\ldots,n\}$ of $G$, the angle
\begin{equation}
  \theta_{a,b} = \arg(U(a,a)) - \arg(U(b,b))
\end{equation}
satisfies
\begin{equation}
  \Bigabs{\theta_{a,b}-\frac{2\pi}{3}} \leq 12\sqrt{\eta}
  \quad\text{or}\quad
  \Bigabs{\theta_{a,b}-\frac{4\pi}{3}} \leq 12\sqrt{\eta}.
\end{equation}

Define sets $S_0$, $S_1$, and $S_2$ as
\begin{equation}
  S_0 = \biggl[\frac{11\pi}{6}, 2\pi\biggr) \cup
    \biggl[0,\frac{\pi}{6}\biggr],
    \quad
    S_1 = \biggl[\frac{\pi}{2},\frac{5\pi}{6}\biggr],
    \quad\text{and}\quad
    S_2 = \biggl[\frac{7\pi}{6},\frac{3\pi}{2}\biggr].
\end{equation}
Because
\begin{equation}
  12n\sqrt{\eta} < \frac{\pi}{6},
\end{equation}
it follows from an iterative application of the argument above that,
for any two \emph{connected} (but not necessarily adjacent) vertices
$a,b\in\{1,\ldots,n\}$ of $G$, exactly one of the following three inclusions
holds:
\begin{equation}
  \label{eq:angle-alternatives}
  \begin{aligned}
    \theta_{a,b} & \in \bigl[2\pi - 12 n\sqrt{\eta}, 2\pi\bigr)
      \cup \bigl[0, 12 n\sqrt{\eta}\bigr] \subseteq S_0,\\
    \theta_{a,b} & \in \biggl[\frac{2\pi}{3} - 12 n\sqrt{\eta},
      \frac{2\pi}{3} + 12 n\sqrt{\eta}\biggr] \subseteq S_1,\\
    \theta_{a,b} & \in \biggl[\frac{4\pi}{3} - 12 n\sqrt{\eta},
      \frac{4\pi}{3} + 12 n\sqrt{\eta}\biggr] \subseteq S_2.
  \end{aligned}
\end{equation}
A 3-coloring of $G$ may therefore be obtained by repeating the following
procedure for each connected component $H$ of $G$:
\begin{enumerate}
\item[1.] Choose an arbitrary vertex $a\in\{1,\ldots,n\}$ of $H$ (or, for
  concreteness, the lowest-numbered vertex of $H$), and assign this vertex the
  color 0 (i.e., set $\varphi(a) = 0$).
\item[2.] For each vertex $b\in\{1,\ldots,n\}$ of $H$, assign $b$ a color as
  follows:
  \begin{enumerate}
  \item[(a)] If $\theta_{a,b} \in S_0$, then set $\varphi(b) = 0$.
  \item[(b)] If $\theta_{a,b} \in S_1$, then set $\varphi(b) = 1$.
  \item[(c)] If $\theta_{a,b} \in S_2$, then set $\varphi(b) = 2$.
  \end{enumerate}
\end{enumerate}
As the angle $\theta_{b,c} = \arg(U(b,b)) - \arg(U(c,c))$ must exceed the width
$\pi/3$ of each set $S_0$, $S_1$, and $S_2$, adjacent vertices cannot be
assigned the same color, and therefore $\varphi$ is a valid 3-coloring of $G$.
This completes the proof of the reduction from 3-coloring to $\text{UQM}$,
aside from the proof of Lemma~\ref{lemma:angle-bound}, which follows.

\begin{proof}[Proof of Lemma~\ref{lemma:angle-bound}]
  It suffices to prove the lemma for $\theta = \arg(\alpha) - \arg(\beta)$, as
  the bound follows for the other two angles by symmetry.

  Observe first that the triangle inequality implies
  \begin{equation}
    \bigabs{ \abs{\alpha+\beta} - \abs{\gamma}} \leq \varepsilon,
  \end{equation}
  and therefore
  \begin{equation}
    1 - 2\varepsilon \leq \abs{\alpha + \beta} \leq 1 + \varepsilon.
  \end{equation}
  It is the case that
  \begin{equation}
    \abs{\alpha + \beta}^2 = \abs{\alpha}^2 + \abs{\beta}^2 + 2 \abs{\alpha}
    \abs{\beta} \cos(\theta),
  \end{equation}
  and therefore
  \begin{equation}
    \cos(\theta) = \frac{\abs{\alpha+\beta}^2
      -\abs{\alpha}^2 -\abs{\beta}^2}{2\abs{\alpha}\abs{\beta}}.
  \end{equation}
  A lower bound on $\cos(\theta)$ may be obtained as follows:
  \begin{equation}
    \label{eq:cosine-lower-bound}
    \cos(\theta)
    = \frac{\abs{\alpha+\beta}^2}{2\abs{\alpha}\abs{\beta}}
    - \frac{1}{2}\biggl(
    \frac{\abs{\alpha}}{\abs{\beta}}+\frac{\abs{\beta}}{\abs{\alpha}}
    \biggr)
    \geq \frac{(1 - 2\varepsilon)^2}{2} - \frac{2+\varepsilon}{2}
    \geq -\frac{1}{2} - \frac{5\varepsilon}{2}.
  \end{equation}
  Here we have used the observation that
  $\abs{\alpha},\abs{\beta}\in[1-\varepsilon,1]$ implies
  \begin{equation}
    \frac{\abs{\alpha}}{\abs{\beta}}+\frac{\abs{\beta}}{\abs{\alpha}}
    \leq (1 - \varepsilon) + \frac{1}{1-\varepsilon} \leq 2 + \varepsilon.
  \end{equation}
  An upper bound on $\cos(\theta)$ is given by
  \begin{equation}
    \label{eq:cosine-upper-bound}
    \cos(\theta)
    = \frac{\abs{\alpha+\beta}^2 - \abs{\alpha}^2 - \abs{\beta}^2}{
      2\abs{\alpha}\abs{\beta}}
    \leq \frac{(1 + \varepsilon)^2 - 2(1 - \varepsilon)^2}{2}
    \leq -\frac{1}{2} + 3\varepsilon.
  \end{equation}
  (Note that the numerator
  $\abs{\alpha+\beta}^2 - \abs{\alpha}^2 - \abs{\beta}^2$ of the first
  fraction in \eqref{eq:cosine-upper-bound} is necessarily non-positive, which
  explains why the denominator of the second fraction is $2$ and not
  $2(1-\varepsilon)^2$.)
  
  Now, because $\cos(\theta)$ is non-positive, it cannot be that
  $\theta \in [0,\pi/2) \cup (3\pi/2,2\pi)$. 
  It therefore suffices to consider the case that $\theta\in[\pi/2,3\pi/2]$.
  We will split this case into two sub-cases, $\theta\in[\pi/2,\pi]$ and
  $\theta\in[\pi,3\pi/2]$, which can be handled by symmetric arguments.
  
  With this in mind, suppose that $\theta\in[\pi/2,\pi]$, and observe that the
  cosine function is convex on the interval $[\pi/2,\pi]$.
  On this interval, the graph of the cosine function therefore lies above the
  tangent line passing through the point $(2\pi/3,-1/2)$, which implies
  \begin{equation}
    \cos(\theta) \geq -\frac{1}{2} + \frac{\sqrt{3}}{2}\biggl(\frac{2\pi}{3} -
    \theta\biggr).
  \end{equation}
  Combining this inequality with \eqref{eq:cosine-upper-bound} yields
  \begin{equation}
    \frac{2\pi}{3} - \theta \leq \frac{6\varepsilon}{\sqrt{3}} \leq 4
    \varepsilon.
  \end{equation}
  Again using convexity, the graph of the cosine function on the interval
  $[2\pi/3,\pi]$ lies below the line segment whose endpoints are
  $(2\pi/3,-1/2)$ and $(\pi,-1)$.
  If $\theta\in [2\pi/3,\pi]$, then it follows that
  \begin{equation}
    \cos(\theta) \leq -\frac{1}{2} -
    \frac{3}{2\pi} \Bigl(\theta - \frac{2\pi}{3}\Bigr),
  \end{equation}
  and therefore by \eqref{eq:cosine-lower-bound} we have
  \begin{equation}
    \theta - \frac{2\pi}{3} \leq \frac{5\pi\varepsilon}{3} \leq 6 \varepsilon.
  \end{equation}
  The same bound is, of course, trivial when $\theta\in[\pi/2, 2\pi/3]$.
  It is therefore the case that
  \begin{equation}
    \Bigabs{\theta-\frac{2\pi}{3}} \leq 6\varepsilon.
  \end{equation}

  A similar argument implies that if $\theta\in[\pi,3\pi/2]$, then
  \begin{equation}
    \Bigabs{\theta-\frac{4\pi}{3}} \leq 6\varepsilon.
  \end{equation}
  which completes the proof.
\end{proof}

\section{Reduction from unitary quadratic optimization to mixed-unitary
  detection}
\label{sec:reduction-to-MUD}

In this section we prove that there exists a polynomial-time Turing reduction
from the \emph{unitary quadratic minimization} problem to the
\emph{mixed-unitary detection} problem:
\begin{equation}
  \label{eq:UQM-to-MUD}
  \text{UQM} \leq^p_T \text{MUD}.
\end{equation}
At the heart of this reduction is a general result due to Liu \cite{Liu2007}
that establishes that there exists a polynomial-time Turing reduction from the
\emph{weak optimization} problem to the \emph{weak membership} problem for
certain convex sets and problem parameterizations.
These problems and Liu's reduction (but not the specifics of the reduction
itself or the proof that it is correct) are described in the first subsection
that follows, and the subsequent subsections connect these problems and Liu's
reduction to the reduction \eqref{eq:UQM-to-MUD}.

\subsection{Weak optimization, weak membership, and Liu's reduction}

In order to define the \emph{weak optimization} and \emph{weak membership}
problems, and to explain Liu's reduction from \emph{weak optimization} to
\emph{weak membership}, a couple of definitions will be required.
The first definition simply establishes some convenient notation.
\begin{definition}
  Let $N$ be a positive integer and let $\delta\geq 0$ be a non-negative real
  number.
  For every vector $x\in\real^N$, the (closed) ball of radius $\delta$ around
  $x$ is defined as
  \begin{equation}
    \B_N(x,\delta) = \bigl\{y\in\real^N\,:\,\norm{y-x}\leq \delta\bigr\},
  \end{equation}
  and for every set $\A\subset\real^N$, one defines
  \begin{equation}
    \B_N(\A,\delta) = \bigcup_{x\in\A}\B_N(x,\delta).
  \end{equation}
\end{definition}

The second definition is one for a \emph{polynomially bounded} sequence of
convex sets.
Intuitively speaking, one should view the sets in such a sequence as
corresponding in some way to possible input lengths in a computational
problem.
The term \emph{polynomially bounded} refers to both a polynomial upper-bound on
the norm of every element in each set and to an inverse polynomial lower bound
on the size of a ball around 0 that is fully contained within each set.

\begin{definition}
  Let $\K_N\subset\real^N$ be a compact, convex set for each positive integer
  $N$.
  The collection $\{\K_1,\K_2,\ldots\}$ is \emph{polynomially bounded} if there
  exists a polynomial $p$ with the property that
  \begin{equation}
    \B_N(0,1/p(N)) \subseteq \K_N
    \quad\text{and}\quad
    \K_N \subseteq \B_N(0,p(N))
  \end{equation}
  for every positive integer $N$.
\end{definition}

\begin{remark}
  It is common that a somewhat more general definition is used in place of the
  one just given, where the smaller ball that is contained in each $\K_N$ need
  not be centered around 0, and where it is only the ratio of the radii of the
  two balls that needs to be polynomially bounded---but because the simpler
  definition above is sufficient for our needs, we adopt it rather than the
  more general definition.
\end{remark}

We are now ready to define the \emph{weak optimization} and
\emph{weak membership} problems, which are variants of standard problems
in the analysis of geometric algorithms \cite{GroetschelLS1988}.
Both are defined with respect to a collection $\{\K_1,\K_2,\ldots\}$ of
compact, convex sets of the sort considered in the previous definition.
(The problem definitions themselves do not require these collections to be
polynomially bounded, but Liu's result will require this assumption.)

\begin{definition}[Weak membership and weak optimization]
  Let $\K_N\subset\real^N$ be a compact, convex set for each positive integer
  $N$ and let $\K = \{\K_1,\K_2,\ldots\}$.
  \begin{enumerate}
  \item[1.]
    The \emph{weak membership} promise problem $\text{WMEM}(\K)$ for $\K$
    is as follows:\vspace{-2mm}
    
    \noindent
    \begin{longtable}{@{}lp{5.2in}@{}}
      Input: & A vector $x\in\real^N$ and the unary representation $0^m$ of a
      positive integer $m$.\\[1mm]
      Yes: & $\B_N(x,1/m) \subseteq \K_N$.\\[1mm]
      No: & $\B_N(x,1/m) \cap \K_N = \varnothing$.
    \end{longtable}
    \vspace{-3mm}
  \item[2.]
    The \emph{weak optimization} promise problem $\text{WOPT}(\K)$ for $\K$
    is as follows:\vspace{-2mm}
    
    \noindent
    \begin{longtable}{@{}lp{5.2in}@{}}
      Input: & A vector $u\in\real^N$ with $\norm{u} \leq 1$, a real number
      $\beta$, and the unary representation $0^m$ of a positive
      integer~$m$.\\[1mm]
      Yes: & There exists a vector $x\in\real^N$ such that
      $\B_N(x,1/m) \subseteq \K_N$ and $\ip{u}{x} \leq \beta$.\\[1mm]
      No: & For every vector $x\in\B_N(\K_N,1/m)$ it is the case that
      $\ip{u}{x} \geq \beta + 1/m$.
    \end{longtable}
  \end{enumerate}
\end{definition}

These problems are referred to as \emph{weak} versions of membership testing
and optimization because the promises effectively make the problems easier than
they might otherwise be.
That is, in the case of weak membership testing, any point within a distance
$1/m$ of the boundary of the corresponding convex set is viewed as a ``don't
care'' input, as is any input to the \emph{weak optimization} problem for which
the objective value $\beta$ is (informally speaking) neither
``easily achievable'' or ``far from achievable.''
In contrast, \emph{strong} variants of these problems, in which the ``don't
care'' inputs just described are valid inputs, are often extremely hard for
reasons that are more closely connected with limitations of finite-precision
real number computations than with the structure of the convex sets being
considered.
Variants of these problems in which $m$ is input in binary rather than unary
notation are also commonly referred to as \emph{weak membership testing} and
\emph{weak optimization}, but the inverse polynomial bound obtained by taking
$m$ to be input in unary is an essential feature of Liu's result and is
required for our purposes.

Finally, we may now state the result due to Liu that forms the heart of
the reduction from \emph{unitary quadratic minimization} to
\emph{mixed-unitary detection}.

\begin{theorem}[Liu]
  \label{theorem:Liu}
  For every polynomially bounded collection $\K = \{\K_1,\K_2,\ldots\}$
  of compact, convex sets, it is the case that
  \begin{equation}
    \label{eq:Liu-reduction}
    \textup{WOPT}(\K) \leq^p_T \textup{WMEM}(\K).
  \end{equation}
  Moreover, there exists a strong polynomial-time Turing reduction that
  establishes this relation.
\end{theorem}

\begin{remark}
  Liu actually proved this theorem for a slightly more restricted version of
  $\text{WOPT}(\K)$ in which the vector $u$ must satisfy $\norm{u} = 1$ rather
  than $\norm{u} \leq 1$.
  The benefit of adopting the definition with the inequality $\norm{u} \leq 1$
  rather than the equality $\norm{u} = 1$ is that it allows us to largely
  circumvent precision issues that arise when taking square roots of rational
  numbers.
  Fortunately, Liu's theorem still holds for the less restricted variant of
  $\text{WOPT}(\K)$, as it has been defined above, as there exists a strong
  polynomial-time mapping reduction from this problem to its more restricted
  variant, under the assumption that there exists a polynomial $p$ for which
  $\K_N\subseteq\B_N(0,p(N))$ for every positive integer $N$ (which, of course,
  is the case when $\K$ is polynomially bounded).

  We will now argue that this is so.
  Consider a reduction that transforms a given instance $(u,\gamma,0^m)$ of
  $\text{WOPT}(\K)$, where $u\in\real^N$ satisfies $\norm{u}\leq 1$, to an
  instance $(v,\delta,0^{4m})$ of the equality-restricted variant of
  $\text{WOPT}(\K)$, where $v$ and $\delta$ are as follows:
  \begin{enumerate}
  \item[1.]
    The vector $v\in\real^N$ is a unit vector satisfying
    \begin{equation}
      \biggnorm{ v - \frac{u}{\norm{u}}} \leq \frac{1}{4 m (p(N)+1)}.
    \end{equation}
  \item[2.]
    The number $\delta$ satisfies
    \begin{equation}
      \frac{\gamma}{\norm{u}} + \frac{1}{4m}
      \leq \delta \leq
      \frac{\gamma}{\norm{u}} + \frac{1}{2m}.
    \end{equation}
  \end{enumerate}

  In the case that $(u,\gamma,0^m)$ is a yes-instance of $\text{WOPT}(\K)$,
  one has that there exists a vector $x\in\real^N$ such that
  $\B_N(x,1/m)\subseteq\K_N$ and $\ip{u}{x} \leq \gamma$.
  The same vector $x$ trivially satisfies $\B_N(x,1/(4m))\subseteq\K_N$, as
  well as
  \begin{equation}
    \ip{v}{x} = \biggip{\frac{u}{\norm{u}}}{x}
    + \biggip{v - \frac{u}{\norm{u}}}{x}
    \leq \frac{\gamma}{\norm{u}} + \frac{\norm{x}}{4 m (p(N)+1)}
    \leq \delta.
  \end{equation}
  The instance $(v,\delta,0^{4m})$ is therefore a yes-instance of the
  equality-restricted variant of $\text{WOPT}(\K)$.

  In the case that $(u,\gamma,0^m)$ is a no-instance of $\text{WOPT}(\K)$,
  every vector $x\in\B_N(\K_N,1/m)$ satisfies $\ip{u}{x} \geq \gamma + 1/m$,
  and therefore also satisfies
  \begin{equation}
    \begin{multlined}
      \ip{v}{x}
      = \biggip{\frac{u}{\norm{u}}}{x}
      + \biggip{v-\frac{u}{\norm{u}}}{x}\\[1mm]
      \geq
      \frac{\gamma}{\norm{u}}+\frac{1}{m \norm{u}}-\frac{\norm{x}}{4m(p(N)+1)}
      \geq
      \frac{\gamma}{\norm{u}}+\frac{3}{4m} \geq \delta +
      \frac{1}{4m}.
    \end{multlined}
  \end{equation}
  Of course this is therefore true for all $x\in\B_N(\K_N,1/(4m))$, so
  $(v,\delta,0^{4m})$ is a no-instance of the equality-restricted variant of
  $\text{WOPT}(\K)$.

  Given $(u,\gamma,0^m)$, one can compute $(v,\delta,0^{4m})$ in polynomial
  time by performing the required arithmetic computations to inverse-polynomial
  accuracy.
  If it is the case that $\gamma$ and the entries of $u$ are given by ratios of
  integers that are bounded in absolute value by some polynomial in $N$, then
  the numerators and denominators of $\delta$ and the entries of $v$ will also
  be polynomially bounded in absolute value, and therefore this is a strong
  polynomial-time mapping reduction.  
\end{remark}

\subsection{Full-dimensional real convex sets for mixed-unitary optimization}

The \emph{weak optimization} and \emph{weak membership} problems are concerned
with convex subsets of $\real^N$, for different choices of $N$, and the
assumption that $\K = \{\K_1,\K_2,\ldots\}$ is a polynomially bounded
collection of compact, convex sets implies that these sets are
full-dimensional.
On the other hand, the \emph{unitary quadratic minimization} and
\emph{mixed-unitary detection} problems are concerned with complex operators,
and moreover (as will shortly become clear), these problems are most naturally
connected with affine subspaces of vector spaces that do not have full
dimension.
In this section we consider a particular family $\K = \{\K_1,\K_2,\ldots\}$
that will allow for a translation from \emph{unitary quadratic minimization} to
\emph{weak optimization} and from \emph{weak membership} to
\emph{mixed-unitary detection}.
It is also proved that $\K$ is polynomially bounded, so that Liu's reduction
holds for this choice of~$\K$.

To begin, for a given positive integer $n\geq 2$, consider the space of all
$n\times n$ traceless Hermitian operators, which is a real vector space of
dimension $n^2-1$.
We will require an orthogonal basis for this space, and one reasonable choice
for such a basis is given by the generalized Gell Mann operators.
Specifically, let $G_1,\ldots,G_{n^2-1}$ denote the elements of
$\Herm(\complex^n)$ obtained by taking the natural ordering suggested by the
following list:
\begin{enumerate}
\item[1.] The first $\binom{n}{2}$ of these operators are
  $E_{j,k} + E_{k,j}$ for $1\leq j < k \leq n$.
\item[2.] The next $\binom{n}{2}$ of these operators are
  $i E_{j,k} - i E_{k,j}$ for $1\leq j < k \leq n$.
\item[3.] The last $n-1$ of these operators are
  \begin{equation}
    \sum_{j = 1}^k E_{j,j} - k E_{k+1,k+1}
  \end{equation}
  for $k = 1,\ldots,n-1$.
\end{enumerate}
\noindent
It will be convenient later to make use of the observation that
$1\leq \norm{G_j}_2 \leq n$ for all $j\in\{1,\ldots,n^2-1\}$.

Let us now define $N = (n^2-1)^2$, which is to be viewed hereafter as a
function of whatever value of $n\geq 2$ is under consideration.
Let $H_1,\ldots,H_N\in\Herm(\complex^n\otimes\complex^n)$ be the operators
obtained by tensoring together the operators $G_1,\ldots,G_{n^2-1}$ in all
possible pairs:
\begin{equation}
  H_1 = G_1\otimes G_1, \quad
  H_2 = G_1\otimes G_2, \quad\ldots,\quad
  H_N = G_{n^2-1}\otimes G_{n^2-1}.
\end{equation}
The operators $H_1,\ldots,H_N$ represent an orthogonal basis for the real
vector space
\begin{equation}
  \V_n = \bigl\{ X\in\Herm(\complex^n\otimes\complex^n)\,:\,
  \bigl(\textup{Tr}\otimes\I_{\Lin(\complex^n)}\bigr)(X) = 0,\;
  \bigl(\I_{\Lin(\complex^n)}\otimes\textup{Tr}\bigr)(X) = 0\bigr\}.
\end{equation}
The relevance of this space is that the smallest real affine subspace of
$\Herm(\complex^n\otimes\complex^n)$ that contains $J(\Phi)$ for every
mixed-unitary channel of the form
$\Phi:\Lin(\complex^n)\rightarrow\Lin(\complex^n)$ is equal to
\begin{equation}
  \V_n + \frac{\I_n\otimes\I_n}{n}.
\end{equation}
Note that
$1\leq \norm{H_j}_2 \leq n^2 < N$ for every $j\in\{1,\ldots,N\}$.

Next, consider the affine linear mapping
$\varphi_n:\real^N \rightarrow \Herm(\complex^n\otimes\complex^n)$
given by
\begin{equation}
  \varphi_n(x) = x(1) H_1 + \cdots + x(N) H_N +
  \frac{\I_n\otimes\I_n}{n}.
\end{equation}
This function is one-to-one, and as $x\in\real^N$ ranges over all vectors,
$\varphi_n(x)$ ranges over the Choi representations of all trace-preserving,
unital, and Hermitian-preserving maps.

Finally, define
\begin{equation}
  \K_N = \bigl\{x\in\real^N\,:\,\varphi_n(x) = J(\Phi)\;\text{for
    $\Phi\in\Channel(\complex^n)$ mixed-unitary}\},
\end{equation}
define $\K_k = \B_k(0,1)$ for each positive integer $k$ that does not take the
form $(n^2 - 1)^2$ for an integer $n\geq 2$, and let
$\K = \{\K_1,\K_2,\ldots\}$.
The particular choice $\K_k = \B_k(0,1)$ when $k \not= (n^2 - 1)^2$ for any
integer $n\geq 2$ is not really important---it is just a trivial choice of
a set for each such dimension that will allow the reduction to work.
Each $\K_N$ is the preimage of the compact and convex set of mixed-unitary
channels under an affine linear map, from which it follows that $\K_N$ is also
compact and convex.
Of course $\K_k$ is trivially compact and convex when $k\not=(n^2-1)^2$ for
every integer $n\geq 2$.

To prove that $\K$ is polynomially bounded, suppose first that $x\in\K_N$ for
$N = (n^2-1)^2$, so that $\varphi_n(x) = J(\Phi)$ for
$\Phi:\Lin(\complex^n)\rightarrow\Lin(\complex^n)$ a mixed-unitary channel.
Because $\Phi$ is a channel, it is the case that
\begin{equation}
  \norm{J(\Phi)}_2 \leq \norm{J(\Phi)}_1 = \tr(J(\Phi)) = n;
\end{equation}
the inequality follows from the fact that $\norm{X}_2\leq\norm{X}_1$ for all
operators, the first equality follows from the fact that $J(\Phi)$ is positive
semidefinite whenever $\Phi$ is a channel, and the second equality follows from
the fact that $\Phi$ must preserve trace.
Because the operators $H_1,\ldots,H_N$ are orthogonal and traceless (and
therefore orthogonal to $\I_n\otimes\I_n$), we conclude that
\begin{equation}
  \norm{J(\Phi)}_2^2
  = \sum_{j = 1}^N x(j)^2 \norm{H_j}_2^2
  + \biggnorm{\frac{\I_n\otimes\I_n}{n}}_2^2
  = \sum_{j = 1}^N x(j)^2 \norm{H_j}_2^2 + 1.
\end{equation}
As $\norm{H_j}_2\geq 1$ for every $j\in\{1,\ldots,N\}$, it follows that
\begin{equation}
  \norm{x}^2 \leq \sum_{j = 1}^N x(j)^2 \norm{H_j}_2^2
  = \norm{J(\Phi)}_2^2 - 1 \leq n^2 - 1,
\end{equation}
and therefore $\norm{x} < n$.
It has therefore been proved that $\K_N \subseteq \B_N(0,n)$.
Of course, when $k\not=(n^2-1)^2$ for every integer $n\geq 2$, it trivially
holds that $\K_N \subseteq \B_N(0,1)$.

To prove that there exists a ball with inverse polynomial radius within
each set $\K_N$, we will make use of the following theorem, which was proved in
\cite{Watrous2009}.
\begin{theorem}
  \label{theorem:mixed-unitary-ball}
  Let $n$ be a positive integer and let
  $\Phi:\Lin(\complex^n)\rightarrow\Lin(\complex^n)$ be a trace-preserving,
  unital, and Hermitian-preserving map.
  If it is the case that
  \begin{equation}
    \biggnorm{J(\Phi) - \frac{\I_n \otimes \I_n}{n}} \leq
    \frac{1}{n(n^2 - 1)},
  \end{equation}
  then $\Phi$ is a mixed-unitary channel.
\end{theorem}

\noindent
For an arbitrary choice of $x\in\real^N$, the mapping $\Phi$ given by
$J(\Phi) = \varphi_n(x)$ satisfies
\begin{equation}
  \begin{multlined}
    \biggnorm{J(\Phi) - \frac{\I_n\otimes\I_n}{n}}^2
    \leq \biggnorm{J(\Phi) - \frac{\I_n\otimes\I_n}{n}}_2^2 \\
    = \sum_{j=1}^N x(j)^2 \norm{H_j}_2^2
    \leq n^4 \sum_{j=1}^N x(j)^2 = n^4 \norm{x}^2.
  \end{multlined}
\end{equation}
Therefore, if
\begin{equation}
  \norm{x} \leq \frac{1}{n^3(n^2-1)},
\end{equation}
then
\begin{equation}
  \biggnorm{J(\Phi) - \frac{\I_n\otimes\I_n}{n}}
  \leq \frac{1}{n(n^2 - 1)},
\end{equation}
so $\Phi$ is mixed-unitary by Theorem~\ref{theorem:mixed-unitary-ball}.
As $N^2 > n^3(n^2 - 1)$, we conclude that
\begin{equation}
  \B_N(0,1/N^2) \subseteq \K_N.
\end{equation}
When $k\not=(n^2-1)^2$ for every integer $n\geq 2$, it trivially holds that
$\B_N(0,1) \subseteq \K_N$.

In conclusion, for all positive integers $k$, it is the case that
\begin{equation}
  \B_k(0,1/k^2) \subseteq \K_k \subseteq \B_k(0,k^2),
\end{equation}
and therefore $\K$ is polynomially bounded.
By Theorem~\ref{theorem:Liu} it therefore follows that
\begin{equation}
  \text{WOPT}(\K) \leq_T^p \text{WMEM}(\K)
\end{equation}
for this choice of $\K$.

\subsection{From unitary quadratic minimization to weak optimization}

In order to prove that $\textup{UQM} \leq_T^p \textup{MUD}$, we will
establish the following chain of reductions:
\begin{equation}
  \label{eq:reduction-chain}
  \text{UQM} \leq_m^p \text{WOPT}(\K) \leq_T^p \text{WMEM}(\K) \leq_m^p
  \text{MUD}.
\end{equation}
The Turing reduction has already been established in the previous subsection,
and in the current subsection we will prove that the first mapping reduction
holds.

To this end, consider an arbitrary instance
\begin{equation}
  (A_1,\ldots,A_k,\alpha,0^m)
\end{equation}
of $\text{UQM}$.
We will first describe the instance of $\text{WOPT}(\K)$ to which each such
instance of $\text{UQM}$ maps, and then we will argue the correctness of the
reduction.

\begin{enumerate}
\item[1.]
  Define an operator $P\in\Herm(\complex^n\otimes\complex^n)$,
  vectors $w,v\in\real^N$, and a real number $\gamma\in\real$ as follows:
  \begin{equation}
    \begin{aligned}
      P & = \sum_{j = 1}^k \vec(A_j) \vec(A_j)^{\ast},\\
      w & = \bigl(\ip{P}{H_1}, \ldots, \ip{P}{H_N}\bigr),\\
      v & = \frac{w}{k N^2},\\
      \gamma & = \alpha - \frac{\tr(P)}{knN^2} + \frac{1}{2kmN^2}.
    \end{aligned}
  \end{equation}
  (The vec mapping refers to the vectorization of an operator
  $A\in\Lin(\complex^n)$:
  \begin{equation}
    \op{vec}(A) = \sum_{1\leq i,j\leq n} A(i,j) e_i\otimes e_j,
  \end{equation}
  which is equivalent to taking the rows of the matrix representation of $A$,
  transposing them to obtain column vectors, and then stacking these column
  vectors on top of one another to form a single vector.)
    
\item[2.]
  Define $u\in\real^N$ as
  \begin{equation}
    u(j) = \frac{\op{trunc}\bigl(8kmnN^3 v(j)\bigr)}{8kmnN^3}
  \end{equation}
  for each $j\in\{1,\ldots,N\}$ and define
  \begin{equation}
    \beta = \frac{\op{trunc}(8kmN^2 \gamma)}{8kmN^2}.
  \end{equation}
  (The truncation function is defined as
  $\text{trunc}(\theta) = \lfloor \theta \rfloor$ and
  $\text{trunc}(-\theta) = - \lfloor \theta \rfloor$ for $\theta\geq 0$,
  so that it always rounds toward zero.)
  
\item[3.]
  The output of the reduction is $(u,\beta,0^r)$ for $r = 8 k m n^4 N^3$.
\end{enumerate}

\noindent
It is evident that $(u,\beta,0^r)$ is polynomial-time computable from
$(A_1,\ldots,A_k,\alpha,0^m)$.
Moreover, under the assumption that $\alpha$ is upper-bounded by a polynomial,
the number $\beta$ and the entries of $u$ can be expressed as ratios of
polynomially bounded integers.
The reduction is therefore a strong polynomial-time mapping reduction.

It remains to argue that if $(A_1,\ldots,A_k,\alpha,0^m)$ is a yes-instance of
$\text{UQM}$ then $(u,\beta,0^r)$ is a yes-instance of $\text{WOPT}(\K)$, and
if $(A_1,\ldots,A_k,\alpha,0^m)$ is a no-instance of $\text{UQM}$ then
$(u,\beta,0^r)$ is a no-instance of $\text{WOPT}(\K)$.
First we note, by the assumption that $\norm{A_j}_2\leq 1$ for each
$j\in\{1,\ldots,k\}$, that $\tr(P) \leq k$.
The norm of $w$ may therefore be upper-bounded,
\begin{equation}
  \norm{w} = \Biggl(\sum_{j=1}^N \abs{\ip{P}{H_j}}^2\Biggr)^{\frac{1}{2}}
  \leq \Biggl(\sum_{j=1}^N \norm{P}_1^2 \norm{H_j}^2\Biggr)^{\frac{1}{2}}
  \leq \sqrt{k^2 N^3} < k N^2,
\end{equation}
which implies $\norm{v} \leq 1$.
As the entries of $u$ are obtained from $v$ by truncations, it is therefore
clear that $\norm{u} \leq 1$.

Next, observe that
\begin{equation}
  \bigip{P}{\vec(U) \vec(U)^{\ast}} = \sum_{j=1}^k \bigabs{\bigip{A_j}{U}}^2
\end{equation}
for every unitary operator $U\in \Unitary(\complex^n)$.
It is evident that
\begin{equation}
  \min_{\Phi\in\MixedUnitary(\complex^n)} \bigip{P}{J(\Phi)} =
  \min_{U\in\Unitary(\complex^n)} \bigip{P}{\vec(U) \vec(U)^{\ast}},
\end{equation}
by virtue of the fact that the function $J(\Phi)\mapsto \ip{P}{J(\Phi)}$ is
linear, the unitary channels are the extreme points of the set of
mixed-unitary channels, and the Choi representation of a unitary channel
$\Phi(X) = U X U^{\ast}$ is
\begin{equation}
  J(\Phi) = \vec(U) \vec(U)^{\ast}.
\end{equation}
It is also the case that
\begin{equation}
  \min_{x\in\K_N} \bigip{P}{\varphi_n(x)} =
  \min_{\Phi\in\MixedUnitary(\complex^n)} \bigip{P}{J(\Phi)}
\end{equation}
because $\varphi_n(x)$ ranges over the set $\MixedUnitary(\complex^n)$ as $x$
ranges over $\K_N$.
Finally, for any choice of a vector $x\in\real^N$, it is the case that
\begin{equation}
  \bigip{P}{\varphi_n(x)} =
  x(1) \ip{P}{H_1} + \cdots + x(N) \ip{P}{H_N} + \frac{\tr(P)}{n}
  = \ip{w}{x} + \frac{\tr(P)}{n}.
\end{equation}
Altogether, this implies that
\begin{equation}
  \min_{x\in\K_N} \ip{v}{x} = 
  \frac{1}{kN^2} \min_{U\in\Unitary(\complex^n)}
  \sum_{j=1}^k \bigabs{\bigip{A_j}{U}}^2
  - \frac{\tr(P)}{knN^2}.
\end{equation}
If $(A_1,\ldots,A_k,\alpha,0^m)$ is a yes-instance of $\text{UQM}$, then it
follows that
\begin{equation}
  \min_{x\in\K_N} \ip{v}{x} \leq \frac{\alpha}{kN^2} - \frac{\tr(P)}{knN^2}
  = \gamma - \frac{1}{2kmN^2},
\end{equation}
while if $(A_1,\ldots,A_k,\alpha,0^m)$ is a no-instance of $\text{UQM}$, then
\begin{equation}
  \min_{x\in\K_N} \ip{v}{x} \geq \frac{\alpha}{kN^2} - \frac{\tr(P)}{knN^2}
  + \frac{1}{k m N^2} = \gamma + \frac{1}{2kmN^2}.
\end{equation}

We now turn to the vector $u$ and the real number $\beta$, which may be viewed
as approximations of $v$ and $\gamma$, respectively.
In particular,
\begin{equation}
  \norm{u-v} \leq \norm{u-v}_1 \leq \frac{1}{8kmnN^2}
\end{equation}
and 
\begin{equation}
  \abs{\beta - \gamma} \leq \frac{1}{8kmN^2}.
\end{equation}
We have already proved that $\norm{x}\leq n$ for every $x\in\K_N$, and this
implies that
\begin{equation}
  \bigabs{\ip{u}{x} - \ip{v}{x}} \leq \frac{1}{8kmN^2}
\end{equation}
for all $x\in\K_N$.
We conclude that
\begin{equation}
  \min_{x\in\K_N} \ip{u}{x} \leq \beta - \frac{1}{4kmN^2}
  \quad\text{or}\quad
  \min_{x\in\K_N} \ip{u}{x} \geq \beta + \frac{1}{4kmN^2},
\end{equation}
depending on whether $(A_1,\ldots,A_k,\alpha,0^m)$ is a yes- or no-instance of
$\text{UQM}$, respectively.

Now, observe that for every $x\in\K_N$ and $\varepsilon \in [0,1]$, it is
the case that
\begin{equation}
  \B_N\biggl((1-\varepsilon)x,\frac{\varepsilon}{n(n^2-1)N}\biggr)
  \subseteq\K_N.
\end{equation}
This is a consequence of Theorem~\ref{theorem:mixed-unitary-ball}, for if
$z\in\real^N$ satisfies
\begin{equation}
  \norm{z} \leq \frac{1}{n(n^2-1)N},
\end{equation}
then
\begin{equation}
  \biggnorm{\varphi_n(z) - \frac{\I_n\otimes\I_n}{n}}
  \leq 
  \biggnorm{\varphi_n(z) - \frac{\I_n\otimes\I_n}{n}}_2
  \leq N\, \norm{z}
  \leq 
  \frac{1}{n(n^2-1)},
\end{equation}
so that $\varphi_n(z) \in \text{MU}(\complex^n)$, and therefore
$(1-\varepsilon)\varphi_n(x) + \varepsilon \varphi_n(z)$ is a convex
combination of mixed-unitary channels.
In particular, for
\begin{equation}
  \varepsilon = \frac{1}{8kmnN^2}
\end{equation}
we have that
\begin{equation}
  \B_N\biggl((1-\varepsilon)x,\frac{1}{r}\biggr) \subseteq\K_N
\end{equation}
for every $x\in\K_N$.
If $(A_1,\ldots,A_k,\alpha,0^m)$ is a yes-instance of $\text{UQM}$, then
there must exist $x\in\K_N$ so that
\begin{equation}
  \ip{u}{x} \leq \beta - \frac{1}{4kmN^2}.
\end{equation}
As
\begin{equation}
  \bigabs{\ip{u}{x} - \ip{u}{(1-\varepsilon)x}} = \varepsilon \abs{\ip{u}{x}}
  < \varepsilon n
\end{equation}
it is the case that
\begin{equation}
  \ip{u}{(1-\varepsilon)x} \leq \beta - \frac{1}{8kmN^2}.
\end{equation}
This implies that $(u,\beta,0^r)$ is a yes-instance of $\text{WOPT}(\K)$.

Finally, if $(A_1,\ldots,A_k,\alpha,0^m)$ is a no-instance of $\text{UQM}$,
then
\begin{equation}
  \ip{u}{x} \geq \beta + \frac{1}{4kmN^2}
\end{equation}
for every $x\in\K_N$.
For every $x\in\B_N(\K_N,1/r)$ we therefore have
\begin{equation}
  \ip{u}{x} \geq \beta + \frac{1}{4kmN^2} - \frac{1}{r}
  \geq \beta + \frac{1}{8kmN^2},
\end{equation}
where the first inequality makes use of the fact that $\norm{u}\leq 1$.
This implies that $(u,\beta,0^r)$ is a no-instance of $\text{WOPT}(\K)$, and
therefore completes the proof that $\text{UQM}\leq_m^p\text{WOPT}(\K)$.

\subsection{From weak membership to mixed-unitary detection}

The remaining reduction in the chain \eqref{eq:reduction-chain} is the
reduction $\textup{WMEM}(\K)\leq_m^p\textup{MUD}$, which we prove in this
subsection.

Before describing the reduction, it will be helpful to observe the
following fact:
if $N = (n^2-1)^2$ for some integer $n\geq 2$, and $y,z\in\real^N$ are
arbitrary vectors, then for the mappings $\Psi$ and $\Xi$ defined by
$J(\Psi) = \varphi_n(y)$ and $J(\Xi) = \varphi_n(z)$, it is the case that
\begin{equation}
  \label{eq:2-norm-bound-on-Psi}
  \norm{J(\Psi) - J(\Xi)}_2^2
  = \sum_{j = 1}^N (y(j) - z(j))^2 \norm{H_j}_2^2 \geq \norm{y-z}^2,
\end{equation}
by virtue of the fact that $\norm{H_j}_2\geq 1$ for every
$j\in\{1,\ldots,N\}$.

Now consider an arbitrary instance $(x,0^m)$ of $\textup{WMEM}(\K)$.
There are two cases to be considered, the first of which is that $x\in\real^N$
for $N = (n^2-1)^2$, where $n\geq 2$ is an integer.
In this case, the first step of the reduction is to compute a vector
$z\in\real^N$ as follows:
\begin{equation}
  z(j) = \frac{\op{trunc}\bigl(2mN x(j)\bigr)}{2mN}.
\end{equation}
The vector $z$ satisfies $\norm{x-z}\leq 1/(2m)$, and is such that every entry
shares the same denominator $2mN$.
(This property will be needed to guarantee that the reduction is strong, and
specifically to avoid a situation in which a polynomial number of rational
numbers, each of which has a polynomially bounded denominator, have an
exponentially large least common denominator.)
The second and final step of the reduction is to output the instance
\begin{equation}
  \bigl(\varphi_n(z),0^{2m}\bigr)
\end{equation}
of $\text{MUD}$.

To prove that this reduction operates correctly for the case being considered,
let $\Phi$ and $\Xi$ be the maps defined by $J(\Phi) = \varphi_n(x)$ and
$J(\Xi) = \varphi_n(z)$, let $\Psi$ be any unital, trace-preserving, and
Hermitian-preserving map that satisfies
\begin{equation}
  \label{eq:Phi-and-Psi-close}
  \norm{J(\Psi) - J(\Xi)}_2 \leq \frac{1}{2m},
\end{equation}
and let $y\in\real^N$ be the unique vector satisfying
$J(\Psi) = \varphi_n(y)$.
By \eqref{eq:2-norm-bound-on-Psi} we have \mbox{$\norm{y-z}\leq 1/(2m)$},
and therefore $\norm{x - y}\leq 1/m$.
If $(x,0^m)$ is a yes-instance of $\textup{WMEM}(\K)$, we therefore have that
$y\in\K_N$, which implies that $\Psi$ is a mixed-unitary channel, and hence
$(J(\Xi),0^{2m}) = (\varphi_n(z),0^{2m})$ is a yes-instance of $\text{MUD}$.
Similarly, if $(x,0^m)$ is a no-instance of $\textup{WMEM}(\K)$, we have that 
$y\not\in\K_N$, which implies that $\Psi$ is not mixed-unitary, and hence
$(J(\Xi),0^{2m}) = (\varphi_n(z),0^{2m})$ is a no-instance of $\text{MUD}$.

In the case that $(x,0^m)$ is an instance of $\textup{WMEM}(\K)$ for which
$x\in\real^k$ for $k$ a positive integer that is not of the form $(n^2 - 1)^2$
for some choice of an integer $n\geq 2$, it is straightforward to decide
whether $(x,0^m)$ is a yes-input or no-input by simply computing the norm of
$x$ numerically to additive error strictly less than $1/m$, then comparing the
result to 1.
In case $(x,0^m)$ is a yes-instance of $\textup{WMEM}(\K)$, the reduction
may output any fixed yes-instance of $\textup{MUD}$, and if $(x,0^m)$ is a
no-instance of $\textup{WMEM}(\K)$, then the reduction may output any fixed
no-instance of $\textup{MUD}$.
It has therefore been been proved that $\textup{WMEM}(\K)\leq_m^p\textup{MUD}$,
which completes the proof that $\text{UQM}\leq_T^p\text{MUD}$.

\section{Conclusion}
\label{sec:conclusion}

We have proved that it is strongly NP-hard, with respect to polynomial-time
Turing reductions, to determine if a given quantum channel is mixed-unitary,
promised that the given channel is not within an inverse-polynomial distance to
the boundary of the set of mixed-unitary channels.
We conclude with a few open problems and directions for future research
relating to this result.
\begin{enumerate}
\item[1.]
  As was suggested in the introduction, an operator
  $X\in\Lin(\complex^n\otimes\complex^n)$ satisfies $X = J(\Phi)$ for a
  mixed-unitary channel $\Phi:\Lin(\complex^n)\rightarrow\Lin(\complex^n)$
  if and only if $X = n \rho$ for $\rho$ being a bipartite quantum state of two
  $n$-dimensional systems that can be expressed as a convex combination of
  maximally entangled pure states:
  \begin{equation}
    \rho = \sum_{k=1}^N p_k u_k u_k^{\ast},
  \end{equation}
  where $u_1,\ldots,u_N\in\complex^n\otimes\complex^n$ satisfy
  \begin{equation}
    \bigl(\textup{Tr}\otimes\I_{\Lin(\complex^n)}\bigr)(u_k u_k^{\ast})
    = \I_n
  \end{equation}
  for each $k\in\{1,\ldots,N\}$.
  Our main result therefore establishes that it is NP-hard to determine whether
  or not a given bipartite quantum state can be expressed as a convex
  combination of maximally entangled pure states.
  
  It is also NP-hard to determine whether or not a given bipartite quantum
  state can be expressed as a convex combination of unentangled pure states
  (i.e., is a separable state)
  \cite{Gurvits2003,Ioannou2007,Gharibian2010,Yu2016},
  and it is interesting that these two extremes represent NP-hard decision
  problems.
  The computational hardness of detecting membership in a variety of other
  convex sets of bipartite (or multipartite) quantum states may also be
  considered.

\item[2.]
  What is the computational difficulty of deciding if a given channel is
  mixed-unitary, given more restrictive promises on the channel's distance from
  the boundary of the set of mixed-unitary channels?
  For example, one may consider the problem in which a given channel is
  promised either to be mixed-unitary or to be at an inverse logarithmic (or
  even constant) distance from the boundary of the mixed-unitary channels.
  We note that the analogous problem for separable states is also open.

  We also note that the two distance promises (inverse logarithmic and
  constant) just suggested are sensitive to the specific choice of a distance
  measure from the boundary of the mixed-unitary channels, in comparison to the
  inverse polynomial distance case we have studied in this paper.
  A variety of distance measures between channels, including the trace-norm,
  2-norm, and spectral-norm distances between their Choi representations, as
  well as the completely bounded trace norm (or diamond norm) distance between
  channels, are all equivalent to one another within a polynomial factor, but
  not within a logarithmic or constant factor.

\item[3.]
  Is the \emph{mixed-unitary detection} problem NP-hard with respect to
  polynomial-time mapping reductions?
  Similar to the previous problem, the analogous problem for separable states is
  also open.
    
\end{enumerate}

\subsection*{Acknowledgments}

This research was supported by Canada's NSERC and the Canadian Institute for
Advanced Research.

\bibliographystyle{halpha}
\bibliography{MixedUnitaryHard}

\end{document}